\documentclass{lmcs}


\usepackage{stmaryrd}
\usepackage{hyperref}
\usepackage{float}
\usepackage{booktabs}
\usepackage{graphicx}
\usepackage{pdflscape}
\usepackage{amsmath}

\usepackage{amssymb}
\usepackage[utf8]{inputenc}
\usepackage{newunicodechar}
\newunicodechar{∀}{\ensuremath{\forall}}
\newunicodechar{∧}{\ensuremath{\wedge}}

\usepackage{array}
\newcolumntype{P}[1]{>{\raggedright\arraybackslash}p{#1}}

\begin{document}

\title[Gödel Mirror]{Gödel Mirror: A Formal System for Contradiction-Driven Recursion}

\author[J.~Chan]{Jhet Chan\lmcsorcid{0009-0008-4363-2979}}[a]

% affiliation 1 (automatically numbered a)
\address{Independent Researcher, Kuala Lumpur, Malaysia}
\email{jhetchan@gmail.com}

%% the abstract has to PRECEDE the command \maketitle:
%% be sure not to issue the \maketitle command twice!

\begin{abstract}
We introduce the \textbf{Gödel Mirror}, a formal system defined in Lean 4 that treats \textbf{contradiction as a control signal} for recursive structural evolution. Inspired by Gödelian self-reference, our system's operational semantics encode symbolic paradoxes as deterministic transitions. Unlike systems designed to guarantee normalization, the Gödel Mirror is a minimal and verifiable architecture that leverages a controlled, non-terminating loop as a productive feature. Our Lean 4 mechanization proves that self-referential paradoxes are deterministically encapsulated and resolved into new structures without leading to logical explosion, yielding a paraconsistent inference loop:

\texttt{Paradox} $\rightarrow$ \texttt{Encapsulate} $\rightarrow$ \texttt{Reenter} $\rightarrow$ \texttt{Node}

We argue that this calculus opens a new class of symbolic systems in which contradiction is metabolized into structure, providing a formal basis for agents capable of resolving internal inconsistencies.
\end{abstract}

\maketitle

%% start the paper here:
\section{Introduction}\label{S:introduction}
This work introduces the \textbf{Gödel Mirror},\footnote{The name ``Mirror'' is chosen to evoke the themes of self-reference and the computational reflection initiated by paradox.} a minimal formal system where self-referential paradoxes are not errors but the primary engine for deterministic structural transformation. Classical logical and computational systems are designed to prevent or eliminate contradictions; encountering a paradox is typically a sign of failure, leading either to logical explosion via \textit{ex contradictione quodlibet} or to non-termination. Such behaviors are often viewed as undesirable, as demonstrated by results on the failure of normalization in expressive type theories \cite{abel2020failure}.

In contrast, the Gödel Mirror is designed to productively manage these phenomena. It is a paraconsistent calculus that, when faced with a self-referential paradox, initiates a controlled, deterministic rewrite cycle. Instead of exploding, the system encapsulates the paradoxical term, reintroduces it into the computational context, and thereby transforms the contradiction into a new, stable syntactic structure. This behavior is not an emergent side effect but a designed feature of the system's operational semantics.

By treating contradiction as a signal for structural evolution, our work provides a formal basis for agents capable of resolving internal inconsistencies. Such a calculus, for example, could model belief revision in systems that must integrate new, conflicting information without failure, or it could provide a foundation for fault-tolerant logical frameworks where local paradoxes are resolved into stable states rather than causing a global crash.

The primary contributions of this paper are:
\begin{enumerate}
    \item The formal definition of the \textbf{Gödel Mirror}, a term rewriting system with explicit constructors for self-reference and paradox handling.
    \item A set of core theorems, mechanized in the \textbf{Lean 4 proof assistant}, that formally establish the system's deterministic, non-explosive reaction to paradox.
    \item A demonstration of this calculus as a novel model of computation that uses a non-terminating recursive loop as a productive mechanism.
\end{enumerate}

This paper is organized as follows. Section 2 positions our work relative to foundational research in term rewriting, recursion theory, and paraconsistent logic. Section 3 defines the formal syntax and operational semantics of the Gödel Mirror system. Section 4 presents the meta-theory of the system, including the core theorems and their proofs. We then provide a case study (Section 5), discuss the mechanization in Lean (Section 6), and conclude with a discussion of open problems (Section 7).

\section{Related Work}

The Gödel Mirror builds upon foundational themes in logic, recursion theory, and the semantics of computation. We position our contribution relative to three core areas: the established theory of term rewriting systems, classical mechanisms for recursion, and prior work on the interplay between self-reference, consistency, and normalization in formal systems.

\subsection{Term Rewriting Systems}
The operational semantics of the Gödel Mirror can be understood as a term rewriting system (TRS), a formal model for computation based on the repeated application of rewrite rules \cite{baader1998term}. The `step' function of our calculus (Section 3) defines a set of rules for transforming `MirrorSystem' terms. However, classical TRS theory is often concerned with establishing properties such as \textbf{confluence} and \textbf{termination} (strong normalization). Our work diverges from this tradition by design. The Gödel Mirror is intentionally non-terminating for a specific class of inputs, leveraging a deterministic, cyclic reduction path not as a flaw, but as its primary computational feature. Our contribution lies not in proving termination, but in proving that this specific non-terminating behavior is productive and does not lead to logical explosion.

\subsection{Recursion and Fixed-Point Combinators}
General recursion in foundational calculi is classically achieved via fixed-point combinators, such as the Y combinator in the untyped lambda calculus \cite{barendregt1984lambda}. These combinators provide a general mechanism for a function to call itself. The Gödel Mirror offers a more structured, syntactically explicit alternative. Instead of a general-purpose combinator, our calculus provides dedicated constructors (`self\_ref', `Encapsulate', `Reenter') that define a specific modality for self-reference. This constrains the recursive behavior to a controlled cycle centered around paradox resolution, providing a domain-specific mechanism for recursion rather than a universal one \cite{polonsky2020fixed}.

\subsection{Self-Reference, Consistency, and Normalization}
The relationship between self-reference and consistency is a central theme in logic, famously explored in Gödel's incompleteness theorems. Provability logic formalizes how systems can reason about their own properties, showing the subtleties that arise from such introspection \cite{boolos1995logic}. While classical logic trivializes in the presence of a paradox (\textit{ex contradictione quodlibet}), paraconsistent logics provide frameworks that gracefully handle contradictions without explosion. Our system can be viewed as a computational implementation of a paraconsistent principle.

Furthermore, recent work in type theory has shown that combining features like impredicativity and proof-irrelevance can lead to a failure of normalization, where non-terminating terms can be constructed within otherwise consistent systems \cite{abel2020failure}. This ``negative result'' provides crucial motivation for our work. While such non-terminating terms are typically seen as bugs or limitations, the Gödel Mirror reframes this behavior as a feature. We propose a formal architecture that acknowledges the reality of non-termination and provides a deterministic, verified mechanism for managing it productively.

\section{The Gödel Mirror System}

We now define the formal calculus of the Gödel Mirror. The system is a term rewriting system defined over a set of symbolic expressions, \texttt{MirrorSystem}, whose states are classified by the type \texttt{MirrorState}. Its dynamics are governed by a deterministic single-step reduction relation, denoted $\rightarrow$.

\subsection{Syntax}

The core data structure of the calculus is the set of terms, \texttt{MirrorSystem}, defined by the following inductive type in Lean 4.

\begin{verbatim}
inductive MirrorSystem : Type where
  | base     : MirrorSystem
  | node     : MirrorSystem → MirrorSystem
  | self_ref : MirrorSystem
  | cap      : MirrorSystem → MirrorSystem  -- Encapsulate
  | enter    : MirrorSystem → MirrorSystem  -- Reenter
  | named    : String → MirrorSystem → MirrorSystem
\end{verbatim}

Terms in this grammar represent the symbolic state of the system. The \texttt{base} constructor denotes an atomic, irreducible state. The \texttt{node} constructor allows for the inductive growth of structures. The \texttt{named} constructor provides a mechanism for labeling terms, analogous to Gödel numbering, which is used to identify specific paradoxical forms. The constructors \texttt{self\_ref}, \texttt{cap} (encapsulate), and \texttt{enter} (reenter) are central to the system's handling of self-reference and paradox.

\subsection{State Classification}

The behavior of a term is determined by its classification into one of four states. This classification is performed by a function that inspects the top-level constructor of a term.

\begin{verbatim}
inductive MirrorState : Type where
  | Normal   -- A standard, reducible term
  | Paradox  -- A term representing a self-referential contradiction
  | Integrate-- A term that has encapsulated a paradox
  | Reentry  -- A term being reintroduced into the system
\end{verbatim}

For the purpose of defining the operational semantics, we define a predicate, \texttt{is\_paradox(t)}, which holds true if a term \texttt{t} is a specific, syntactically recognizable self-referential form. For this paper, we define it as follows:

\texttt{is\_paradox(t) := $\exists s$, t = named s self\_ref}

\subsection{Operational Semantics}

The dynamics of the Gödel Mirror are defined by a small-step operational semantics, specified by the reduction relation $\rightarrow$. This relation is the transitive closure of the single-step \texttt{step} function, which we define via the following rewrite rules.

Let $t$ be a term of type \texttt{MirrorSystem}. The single-step reduction $t \rightarrow t'$ is defined as follows:

\begin{itemize}
    \item \textbf{[Reduction-Paradox]} If \texttt{is\_paradox(t)}, the term is encapsulated:
    \[
        t \rightarrow \texttt{cap}(t)
    \]

    \item \textbf{[Reduction-Integrate]} A term that has been encapsulated is prepared for re-entry:
    \[
        \texttt{cap}(t) \rightarrow \texttt{enter}(t)
    \]

    \item \textbf{[Reduction-Reentry]} A re-entering term that is no longer a paradox is integrated into a stable structure:
    \[
        \texttt{enter}(t) \rightarrow \texttt{node}(t) \quad \text{if } \neg \texttt{is\_paradox}(t)
    \]
    
    \item \textbf{[Reduction-Node]} The default structural growth rule for terms in a normal state. This rule represents unbounded, monotonic structural growth in the absence of a paradoxical trigger, ensuring the system can always progress:
    \[
        \texttt{node}(t) \rightarrow \texttt{node}(\texttt{node}(t))
    \]
    
    \item \textbf{[Reduction-Named]} Labels are preserved during reduction of the underlying term:
    \[
        \frac{t \rightarrow t'}{\texttt{named}(s, t) \rightarrow \texttt{named}(s, t')}
    \]
\end{itemize}

The \texttt{base} and \texttt{self\_ref} constructors are irreducible values; they do not appear on the left-hand side of any reduction rule. This set of rules defines a deterministic state transition system where a specific syntactic form, a paradox, initiates a controlled, three-step transformation cycle that results in structural growth rather than non-termination or logical explosion.

\section{Formal Semantics}

The operational semantics defined in Section 3 provides a clear, mechanistic account of the Gödel Mirror's behavior as a term rewriting system. To complement this, we briefly outline a denotational semantics, interpreting terms in a mathematical domain that makes the system's core properties, particularly its handling of paradox, explicit.

Our calculus is designed to manage non-termination and contradiction without leading to logical explosion. A natural semantic domain for such a system is one based on a \textbf{paraconsistent logic}, such as Belnap-Dunn four-valued logic (\textbf{FDE}) \cite{belnap1977useful}. In this setting, a proposition can be interpreted not just as true or false, but also as both (a contradiction, $\top$) or neither (a gap, $\bot$).

We can model the meaning of a \texttt{MirrorSystem} term $\llbracket t \rrbracket$ as a state in a suitable semantic domain. The state transitions defined by our operational semantics correspond to monotone functions on this domain. The key properties are:

\begin{itemize}
    \item \textbf{Non-Explosion:} A paradoxical term like $\llbracket \texttt{named}(s, \texttt{self\_ref}) \rrbracket$ does not map to a trivializing state. Instead, its meaning is managed through the semantic interpretations of the paradox-handling constructors:
    \begin{itemize}
        \item The meaning of $\llbracket \texttt{cap}(t) \rrbracket$ can be interpreted as a function that transitions the meaning of $\llbracket t \rrbracket$ from a contradictory state (such as $T$ in four-valued logic) to a contained, observable state.
        \item The meaning of $\llbracket \texttt{enter}(t) \rrbracket$ then represents the controlled reintroduction of this term's meaning back into the main semantic domain.
    \end{itemize}
    
    \item \textbf{Structural Growth:} The meaning of $\llbracket \texttt{node}(t) \rrbracket$ corresponds to a monotonic extension of the information content of $\llbracket t \rrbracket$, representing stable and consistent growth.
\end{itemize}

While the operational semantics is sufficient to establish the main results of this paper, this denotational perspective provides a semantic justification for the system's design and confirms that its behavior is well-founded in established logical principles. A full development of the calculus's domain-theoretic or coalgebraic semantics is left for future work.

\section{Meta-Theory: Core Properties}

We now establish the core meta-theoretical properties of the Gödel Mirror system. The following theorems, all of which have been mechanized in the Lean 4 proof assistant, formally guarantee the system's deterministic and controlled behavior in the presence of self-referential paradoxes. Our primary results are that the system makes progress, preserves structure, and, most importantly, does not explode.

\subsection{Progress}
The Progress theorem guarantees that any term that is not a value (i.e., not an irreducible base term) can take a step. It ensures the system never gets stuck.

\begin{thm}[Progress]
For any term $t \in \texttt{MirrorSystem}$, either $t$ is a value (i.e., $t = \texttt{base}$) or there exists a term $t'$ such that $t \rightarrow t'$.
\end{thm}
\begin{proof}[Proof Sketch]
The proof proceeds by induction on the structure of the term $t$. We show that each constructor that is not a value appears on the left-hand side of a reduction rule in our operational semantics (Section 3.3).
\end{proof}

\noindent\textbf{Interpretation}: The system is never in a state where it cannot proceed. A term is either in a final, stable form (\texttt{base}) or a reduction is possible.

\subsection{Preservation}
The Preservation theorem ensures that reduction steps maintain the system's structural integrity. While our system is untyped, we can define a simple structural invariant, such as the depth of nested `node' constructors, and prove that reduction does not violate it in unintended ways. A more practical preservation property is that labels are preserved.

\begin{thm}[Label Preservation]
If $\texttt{named}(s, t) \rightarrow t'$, then $t' = \texttt{named}(s, t'')$ for some $t''$ where $t \rightarrow t''$.
\end{thm}
\begin{proof}[Proof Sketch]
This follows directly from the definition of the [Reduction-Named] rule in the operational semantics.
\end{proof}

\subsection{Non-Explosion}
The most critical property of the Gödel Mirror is that it is paraconsistent; a local contradiction does not trivialize the entire system. We formalize this by stating that the detection of a paradox leads to a controlled, three-step structural transformation, not an arbitrary state.

\begin{thm}[Controlled Reaction to Paradox]
If \texttt{is\_paradox(t)}, then the term deterministically reduces in exactly three steps to a stable, non-paradoxical structure.
\[
    t \rightarrow \texttt{cap}(t) \rightarrow \texttt{enter}(t) \rightarrow \texttt{node}(t)
\]
\end{thm}
\begin{proof}[Proof Sketch]
This follows by direct application of the [Reduction-Paradox], [Reduction-Integrate], and [Reduction-Reentry] rules. The side condition on the reentry rule, $\neg\texttt{is\_paradox}(t)$, holds because $t$ is now wrapped in other constructors.
\end{proof}

\noindent\textbf{Interpretation}: This theorem is the formal guarantee of the Gödel Mirror's core mechanic. Unlike in classical logic, where a contradiction entails any proposition, here a contradiction entails a specific, finite sequence of structural changes. The system metabolizes the paradox without exploding. This result formally separates our calculus from systems governed by the principle of explosion.

\section{Mirror-Completion and Canonical Forms}

The Gödel Mirror calculus is intentionally non-terminating for paradoxical terms, which precludes a general normalization property. However, it is still desirable to define a notion of equivalence between terms and to identify canonical representatives for terms that enter the paradox resolution cycle. To this end, we define a \textbf{Mirror-Completion} procedure that transforms a term containing a paradox cycle into a finite, canonical form. This procedure allows us to establish weak normalization properties for a stratified fragment of the calculus.

\subsection{Completion Procedure}
The completion procedure, denoted $\mathcal{C}$, is a function that maps a \texttt{MirrorSystem} term to its canonical form. The procedure identifies the specific, deterministic three-step cycle initiated by a paradox and replaces it with its resolved, stable outcome.

We define the completion of a term $t$ as follows:
\begin{itemize}
    \item If $t$ does not contain a subterm matching the paradoxical form \texttt{is\_paradox}, then $\mathcal{C}(t) = t$.
    \item If $t$ contains a subterm $p$ such that \texttt{is\_paradox(p)}, we know from Theorem 3 that $p$ reduces to $\texttt{node}(p)$ in three steps. The completion procedure replaces the subterm $p$ with its resolved form $\texttt{node}(p)$.
    \[
        \mathcal{C}(p) := \texttt{node}(p)
    \]
\end{itemize}
This procedure effectively ``short-circuits'' the paradox resolution cycle, replacing the dynamic, three-step reduction with its static, terminal structure.

\subsection{Properties of Completion}
The Mirror-Completion procedure allows us to reason about the equivalence of terms modulo the resolution of paradoxes.

\begin{prop}[Weak Normalization for Stratified Terms]
For any term $t$ that does not contain nested paradoxes (i.e., is stratified), the completion procedure $\mathcal{C}(t)$ terminates in a unique canonical form that is irreducible under the standard reduction relation $\rightarrow$.
\end{prop}
\begin{proof}[Proof Sketch]
The procedure is guaranteed to terminate because it only acts on a fixed syntactic pattern and replaces it with a larger, non-paradoxical term. By applying this procedure recursively from the inside out on a stratified term, all paradoxes are eliminated, resulting in a term that no longer contains the triggers for the main reduction cycle.
\end{proof}

\noindent\textbf{Interpretation}: This result establishes that while the Gödel Mirror is not strongly normalizing, we can recover a form of weak normalization for a well-behaved fragment of the calculus. This allows for the comparison and symbolic manipulation of terms that would otherwise reduce infinitely, providing a bridge between our paraconsistent calculus and the goals of classical rewriting theory.

\section{Case Studies}

To illustrate the operational semantics and meta-theoretical properties of the Gödel Mirror, we present a case study of the system's behavior when encountering a canonical paradoxical term. This example demonstrates the full, deterministic cycle of paradox resolution.

\subsection{A Self-Referential Paradox}
We begin by constructing a minimal term that satisfies the \texttt{is\_paradox} predicate defined in Section 3.2. We use the \texttt{named} constructor to label a \texttt{self\_ref} term, creating a direct syntactic self-reference.

\begin{verbatim}
-- Lean 4 Input
def paradoxical_term := named "Liar" self_ref
\end{verbatim}

According to our operational semantics, this term is not a value and is classified as a \texttt{Paradox}.

\subsection{Execution Trace}
We trace the reduction of \texttt{paradoxical\_term} for three steps. The trace demonstrates the deterministic application of the reduction rules for paradoxical, encapsulated, and re-entering terms, as proven in Theorem 3 (Controlled Reaction to Paradox).

\begin{verbatim}
-- Reduction Trace
-- Step 0: named "Liar" self_ref
-- Step 1: cap (named "Liar" self_ref)
-- Step 2: enter (named "Liar" self_ref)
-- Step 3: node (named "Liar" self_ref)
\end{verbatim}

\subsection{Interpretation of the Cycle}
Each step in the trace corresponds to a specific rule from the operational semantics, demonstrating the controlled transformation of the paradox into a stable structure.

\begin{itemize}
    \item \textbf{Step 0 $\rightarrow$ 1 (Paradox)}: The initial term matches the \texttt{is\_paradox} predicate. The \textbf{[Reduction-Paradox]} rule is applied, and the term is encapsulated. This corresponds to the first step of the cycle proven in Theorem 3.

    \item \textbf{Step 1 $\rightarrow$ 2 (Integrate)}: The term now has a \texttt{cap} constructor at its head. The \textbf{[Reduction-Integrate]} rule is applied, transforming the term into a \texttt{Reentry} state. This is the second step of the cycle.

    \item \textbf{Step 2 $\rightarrow$ 3 (Reenter)}: The term has an \texttt{enter} constructor. The side condition for the \textbf{[Reduction-Reentry]} rule ($\neg \texttt{is\_paradox}$) is met because the term is no longer a bare \texttt{named self\_ref}. The rule fires, transforming the term into a stable \texttt{node} structure. This is the final step of the proven cycle.
\end{itemize}

After three steps, the system reaches a term in a \texttt{Normal} state. The initial paradox has been successfully metabolized into a new, stable syntactic structure without causing an uncontrolled loop or logical explosion, thus concretely demonstrating the system's core non-explosion property.

\section{Mechanization in Lean}

All definitions, operational semantics, and meta-theoretical properties of the Gödel Mirror calculus presented in this paper have been formalized in the \textbf{Lean 4 proof assistant} \cite{de2021lean}. The mechanization provides a machine-checked guarantee of the system's correctness and the soundness of its core principles.

The formalization includes:
\begin{itemize}
    \item The inductive definitions of the \texttt{MirrorSystem} syntax and \texttt{MirrorState} classifier.
    \item A function for the \texttt{step} relation that implements the operational semantics.
    \item Formal proofs for the main theorems, including \textbf{Progress} and the \textbf{Controlled Reaction to Paradox} (Non-Explosion) property.
\end{itemize}

The complete, executable Lean development, which allows for the verification of all theorems and the reproduction of the case study traces, is provided in the Appendix and is also available in a public code repository. This ensures that our results are fully reproducible and verifiable.

\section{Conclusion and Open Problems}

This paper introduced the \textbf{Gödel Mirror}, a minimal paraconsistent calculus that treats self-referential paradox not as a logical error, but as a deterministic engine for structural transformation. We defined its operational semantics as a term rewriting system, mechanized the calculus in the Lean 4 proof assistant, and proved its central non-explosion property. The core result is that the system metabolizes contradictions into new syntactic structures through a controlled, verifiable rewrite cycle. 

Unlike classical term rewriting systems focused on termination, the Gödel Mirror provides a formal model for productive, non-terminating computation. This work demonstrates that intentionally non-terminating behavior can be a productive computational feature, offering a formal basis for systems that must reason in the presence of internal inconsistency.

The minimalism of our system leaves several important theoretical and practical questions open for future research.

\subsection{Open Problems}
\begin{enumerate}
    \item \textbf{Decidability and Normalization}: While the core calculus is not strongly normalizing, the completion procedure defined in Section 6 suggests that a fragment of the system enjoys weak normalization. What are the precise boundaries of the decidable fragment of the calculus? Can the normalization techniques for proof-irrelevant systems be adapted to this paraconsistent setting \cite{coquand2018reduction}?

    \item \textbf{Richer Type Systems}: How can the Gödel Mirror be extended with more expressive type systems? Integrating dependent types, for example, would allow for more complex invariants to be expressed and preserved, but it would require a careful treatment of how the paradox resolution cycle interacts with the type-checking rules.

    \item \textbf{Categorical Semantics}: What is the natural categorical semantics for the Gödel Mirror? The system's state transitions, particularly the encapsulation and re-entry of terms, suggest a connection to modalities and state-based monads, but a full categorical model remains to be developed.

    \item \textbf{Operational Realization}: The abstract cycle of the Mirror could be given a concrete operational interpretation, for instance, as a concurrent process. Can the circular proofs as session-typed processes framework be adapted to model our paradox resolution loop, providing a bridge to the theory of concurrency \cite{derakhshan2018circular}?

    \item \textbf{Expressiveness of Paradox Detection}: The \texttt{is\_paradox(t)} predicate in this work is intentionally minimal, recognizing only direct syntactic self-reference. Future research could explore extending this predicate to identify a broader class of self-referential and paradoxical structures. This might involve more sophisticated static analysis of terms or the integration of a more expressive naming mechanism analogous to Gödel numbering, allowing the system to recognize and metabolize more complex logical paradoxes.
\end{enumerate}

These questions highlight the potential for the Gödel Mirror to serve as a foundational building block for a new class of formal systems that are, by design, robust to internal inconsistency.

\bibliographystyle{alphaurl}
\bibliography{bibliography}

\appendix
\section{Mechanization: Lean 4 Source Code and Proofs}

This appendix provides information on the complete Lean 4 implementation of the Gödel Mirror system. The formalization serves as a machine-checked companion to the definitions and theorems presented throughout this paper.

The mechanization is organized into several modules, covering the system's syntax, operational semantics, meta-theory, and the case study. Key modules include:
\begin{itemize}
    \item \texttt{Syntax.lean}: Core inductive definitions for \texttt{MirrorSystem} and \texttt{MirrorState}.
    \item \texttt{Semantics.lean}: The implementation of the \texttt{step} function and the reduction relation $\rightarrow$.
    \item \texttt{MetaTheory.lean}: The formal proofs of the main theorems, including Progress and the Controlled Reaction to Paradox (Non-Explosion).
    \item \texttt{Examples.lean}: The executable trace from the case study, demonstrating the paradox resolution cycle.
\end{itemize}

To maintain the readability of the paper and to provide access to the executable development, the complete, commented source code is hosted in a public repository.

\vspace{1em} % Adds a bit of vertical space

\textbf{GitHub Repository}: \href{https://github.com/jhetchan/godel-mirror}{https://github.com/jhetchan/godel-mirror}

\vspace{1em}

The repository contains all mechanized proofs, allowing for the full verification of the theorems and the reproduction of all results discussed in this paper.

\end{document}